\documentclass[11pt]{myclass}
\usepackage{times}
\usepackage{euscript}

\usepackage{color,xspace}
\definecolor{MyGray}{rgb}{0.9,0.9,0.9}
\definecolor{MyRed}{rgb}{1.0,0.0,0.0}


\definecolor{MyGreen}{rgb}{0.0,0.8,0.0}
\definecolor{MyBlue}{rgb}{0.0,0.0,1.0}

\long\def\symbolfootnote[#1]#2{\begingroup
\def\thefootnote{\fnsymbol{footnote}}\footnote[#1]{#2}\endgroup} 

\usepackage{picins}

\newlength{\ppicwd}
\newcommand\PPicCap[3]{
    \begin{minipage}{#1}
        \hspace{-0.4cm}
        \settowidth{\ppicwd}{\includegraphics{#2}}  
        \makebox[0cm][l]{\includegraphics[scale=.6]{#2}}

       \vspace{-7\lineskip}
       \begin{minipage}{\ppicwd}
         	#3
       \end{minipage}
    \end{minipage}}

\newcommand{\la}{\lambda}
\newcommand{\ka}{\kappa}

\newcommand{\eps}{\varepsilon}
\newcommand{\Eu}[1]{\ensuremath{\EuScript{#1}}}
\renewcommand{\b}[1]{\ensuremath{\mathbb{#1}}}

\newcommand{\realn}[1]{\reals^{#1}}

\newcommand{\hboln}[1]{\mathbb{H}^{#1}}
\newcommand{\PD}[1]{\text{P}(#1)}
\newcommand{\Sym}[1]{\text{S}(#1)}
\newcommand{\GL}[1]{\text{GL}(#1)}
\newcommand{\SO}[1]{\text{SO}(#1)}

\newcommand{\CAT}[1]{\text{CAT}(#1)}

\newcommand{\inner}[2]{\left<#1,#2\right>}

\newcommand{\inv}[1]{#1^{-1}}
\newcommand{\trn}[1]{#1^{T}}
\newcommand{\itr}[1]{#1^{-T}}

\newcommand{\tr}{\mathop{\mathrm{tr}}}

\newcommand{\GA}{\mathop{\mathrm{GA}}}
\newcommand{\diam}[1]{\ensuremath \textsf{diam(}#1\textsf{)}}
\newcommand{\symnorm}[1]{\sqrt{\tr(#1^2)}}

\newcommand{\actx}[2]{#1#2\trn{#1}}
\newcommand{\actxt}[2]{\trn{#1}#2#1}
\newcommand{\actxi}[2]{\inv{#1}#2\itr{#1}}

\newcommand{\mattwosym}[3]{\begin{pmatrix}#1 & #3 \\ #3 & #2\end{pmatrix}}
\newcommand{\mattwobsym}[3]{\begin{pmatrix}#1 & #3 \\ \trn{#3} & #2\end{pmatrix}}

\newcommand{\bdry}[1]{\partial {#1}} 
\newcommand{\intr}[1]{\text{int}\ {#1}}

\newcommand{\hb}[2]{B_{#2}(#1)}
\newcommand{\hs}[2]{S_{#2}(#1)}

\newcommand{\ext}[2]{E_{#1}(#2)} 

\newcommand{\ch}[1]{\mathcal{C}(#1)} 
\newcommand{\bh}[1]{\mathcal{B}(#1)} 
\newcommand{\ebh}[1]{\mathcal{B}_{\eps,q}(#1)} 
\newcommand{\bch}[1]{\bdry{\ch{#1}}}
\newcommand{\bbh}[1]{\bdry{\bh{#1}}}
\newcommand{\ich}[1]{\intr{\ch{#1}}}

\title{Computing Hulls And Centerpoints In Positive Definite Space\thanks{This research was
  supported in part by NSF SGER-0841185 and a subaward to the University of Utah under NSF award 0937060 to the Computing Research Association.}}
\author{\normalsize P. Thomas Fletcher \\ {\small\textsl{fletcher@sci.utah.edu}} \and
  \normalsize John Moeller \\ {\small \textsl{moeller@cs.utah.edu}} \and \normalsize Jeff M. Phillips
  \\ {\small\textsl{jeffp@cs.utah.edu}}\and  \normalsize  Suresh Venkatasubramanian \\
  {\small \textsl{suresh@cs.utah.edu}}}
\date{}
\begin{document}
\maketitle

\begin{abstract}

In this paper, we present algorithms for computing approximate hulls and centerpoints for collections of matrices in positive definite space. There are many applications where the data under consideration, rather than being points in a Euclidean space, are positive definite (p.d.) matrices. These applications include diffusion tensor imaging in the brain, elasticity analysis in mechanical engineering, and the theory of kernel maps in machine learning. Our work centers around the notion of a \emph{horoball}: the limit of a ball fixed at one point whose radius goes to infinity. Horoballs possess many (though not all) of the properties of halfspaces; in particular, they lack a strong separation theorem where two horoballs can completely partition the space. In spite of this, we show that we can compute an approximate ``horoball hull'' that strictly contains the actual convex hull. This approximate hull also preserves geodesic extents, which is a result of independent value: an immediate corollary is that we can approximately solve problems like the diameter and width in positive definite space. We also use horoballs to show existence of and compute approximate robust centerpoints in positive definite space, via the horoball-equivalent of the notion of depth. 

\end{abstract}

\thispagestyle{empty}
\newpage

\pagenumbering{arabic}

\section{Introduction}
\label{sec:introduction}
There are many application areas where the basic objects of interest, rather than points in Euclidean space, are symmetric positive-definite $n\times n$ matrices (denoted by $\PD n$). 
In diffusion tensor imaging~\cite{PJBasser01011994}, matrices in $\PD 3$ model the flow of water at each voxel of a brain scan. 
In mechanical engineering~\cite{Cowin1992}, stress tensors are modeled as elements of  $\PD 6$. 
Kernel matrices in machine learning are elements of $\PD n$~\cite{Shawe-Taylor2004}.

In all these areas, a problem of great interest is the analysis~\cite{Fletcher2004,Fletcher2009} of collections of such matrices (finding central points, clustering, doing regression). For all of these problems, we need the same kinds of geometric tools available to us in Euclidean space, including basic structures like halfspaces, convex hulls, Voronoi diagrams, various notions of centers, and the like. $\PD n$ is non-Euclidean; in particular, it is negatively (and variably) curved, which poses fundamental problems for the design of geometric algorithms. This is in contrast to hyperbolic space (which has constant curvature of $-1$), in which many standard geometric algorithms carry over. 

In this paper, we develop a number of basic tools for manipulating positive definite space, with a focus on applications in data analysis.

\subsection{Our Work}
\label{sec:our-work}

\paragraph*{Horoballs.} A main technical contribution of this work is the use of \emph{horoballs} as generalization of halfspaces. Suppose we allow a  ball to grow to infinite radius while always touching a fixed point on its boundary. In Euclidean space, this construction yields a halfspace; in general Cartan-Hadamard manifolds (of which $\PD n$ is a special case), this construction yields a horoball. Because of the curvature of space, horoballs are not flat and the complement of a horoball is not a horoball. However, we show that these objects can be effectively used as proxies for halfspaces, allowing us to define a number of different geometric structures in $\PD n$. 

\paragraph*{Ball Hulls.} The first structure we study is the convex hull. Apart from its importance as a fundamental primitive in computational geometry, the convex hull also provides a compact description of the boundary of a data set, can be used to define the \emph{center} of a data set (via the notion of \emph{convex hull peeling depth}~\cite{Shamos1976,Barnett1976}), and also captures extremal properties of a data set like its diameter, width and bounding volume (even in its approximate form~\cite{Agarwal2004}). 

The convex hull of a set of points in $\PD n$ can be  naturally defined as the intersection of all convex sets containing the points. Alternatively, it can be defined as the set of all points that are ``convex combinations'' (in a geodesic sense) of the input points. A significant obstacle to the convex hull in $\PD n$ is that it is not even known whether the convex hull of a finite collection of points in $\PD n$ can be represented finitely~\cite{Berger2007}. 

Another approach to defining the convex hull is via halfspaces: we can define the convex hull in Euclidean space as the intersection of all halfspaces that contain all the points. Unfortunately, even this notion fails to generalize: the relevant structures are called \emph{totally geodesic submanifolds}, and we cannot guarantee that any set of $d+1$ points admits such a submanifold passing through them. 

Our main technical contribution here is a generalization of the convex hull called the \emph{ball hull} that is based on the relationship between horoballs and halfplanes. The ball hull is the intersection of all horoballs that contain the input points.  Although the ball hull itself might require an infinite number of balls to define it, it is closed, it can be approximated efficiently, it is identical to the convex hull in Euclidean
space, and it always contains the convex hull in $\PD n$. In the process of proving this result, we
also develop a generalized notion of \emph{extent}~\cite{Agarwal2004} in positive definite space that might be of independent interest for other analysis problems. 
    
\paragraph*{Centerpoints.}
One important motivation for studying collections of points in positive definite space is to compute measures of centrality (or \emph{mean shapes})~\cite{Fletcher2009}. A robust centerpoint can be obtained by finding a point of maximum (halfspace) depth among a collection of points. We first prove, using a generalization of Helly's theorem to negatively curved spaces, that for any set of points in $\PD n$, there exists a point of large depth, where depth is defined in terms of horoballs. We then develop an algorithm to compute an approximation to such a point, using an LP-type framework. The point we compute is a geometric approximation: it does not approximate the depth of the optimal point, but is guaranteed to be close to such a point.
 
\subsection{Related Work}
\label{sec:related-work}

The mathematics of Riemannian manifolds, Cartan-Hadamard manifolds and $\PD n$ is well-understood: the book by Bridson and Haefliger~\cite{Bridson2009} is an invaluable reference on metric spaces of nonpositive curvature, and Bhatia~\cite{Bhatia2006} provides a detailed study of $\PD n$ in particular. However, there are many fewer algorithmic results for problems in these spaces. To the best of our knowledge, the only prior work on algorithms for positive definite space are the work by Moakher~\cite{M-B-} on mean shapes in positive definite space, and papers by Fletcher and Joshi~\cite{Fletcher2004} on doing principal geodesic analysis in symmetric spaces, and the robust median algorithms of Fletcher \emph{et al}~\cite{Fletcher2009} for general manifolds (including $\PD n$ and $\SO n$). 

Geometric algorithms in hyperbolic space are much more tractable.  The Poincar\'{e} and Klein models of hyperbolic space preserve different properties of Euclidean space, and many algorithm carry over directly with no modifications. Leibon and Letscher~\cite{Leibon2000} were the first to study basic geometric primitives in general Riemannian manifolds, constructing Voronoi diagrams and Delaunay triangulations for sufficiently dense point sets in these spaces. Eppstein~\cite{Eppstein2009} described hierarchical clustering algorithms in hyperbolic space.  Krauthgamer and Lee~\cite{Krauthgamer2006} studied the nearest neighbor problem for points in $\delta$-hyperbolic space; these spaces are a combinatorial generalization of negatively curved space and are characterized by global, rather than local, definitions of curvature. Chepoi \emph{et al}~\cite{Chepoi2008,Chepoi2007} advanced this line of research, providing algorithms for computing the diameter and minimum enclosing ball of collections of points in $\delta$-hyperbolic space.

\section{Preliminaries}
\label{sec:preliminaries}

$\PD{n}$ is the set of symmetric positive-definite real matrices. It is a Riemannian metric space with tangent space at point $p$ equal to $\Sym n$, the vector space of symmetric matrices with inner product $\inner{A}{B}_p=\tr(\inv{p}A\inv{p}B)$. The $\exp$ map, $\exp_p:\Sym{n}\to\PD{n}$ is defined $\exp_p(tA)=c(t)=p e^{tpA}$, where $c(t)$ is the geodesic with unit tangent $A$ and $c(0)=p$. For simplicity, we often assume that $p=I$ so $\exp_I(tA)=e^{tA}$.
The $\log$ map, $\log_p:\PD{n}\to\Sym{n}$, indicates direction and distance and is the inverse of $\exp_p$. The metric $d(p,q)=\|\log_p(q)\|=\symnorm{\log(\inv{p}q)}$.

\paragraph{Convex Hulls in $\PD n$.}
$\PD n$ is an example of a proper $\CAT 0$
space~\cite[II.10]{Bridson2009}, and as such admits a well-defined notion of convexity, in which metric balls are convex. We can define the convex hull $\ch X$ of a set of points $X$ as the smallest convex set that contains the points. This hull can be realized as the limit of an iterative procedure where we draw all geodesics between data points, add all the new points to the set, and repeat.  

\begin{lemma}[\cite{Bhatia2006}]\label{conv-hull-is-geo-union}
If $X_0 = X$ and $X_{i+1} = \bigcup_{a,b\in X_i}[a,b]$, then $\ch X = \bigcup_{i=0}^\infty X_i$.
\end{lemma}
    
     \begin{proof}
       We will use the notation $X_\infty = \bigcup_{i=0}^\infty X_i$. It is easy to demonstrate by straightforward induction that $X_\infty$ is contained in any convex set that contains $X$. Therefore $\ch X\supseteq X_\infty$.    
      
       We also know that if $p,q\in X_\infty$ there must be some $m$ for which $p,q\in X_m$, since $X_\infty$ is the nested union of $X_i$. Then $[p,q]\subset X_{m+1}\subset X_\infty$. This means that $X_\infty$ is convex, so $\ch X\subseteq X_\infty$.
     \end{proof}

    Berger~\cite{Berger2007} notes that it is unknown whether the convex hull of three points is in general closed, and the standing conjecture is that it is not. The above lemma bears this out, as it is an infinite union of closed sets, which in general is not closed. These facts present a significant barrier to the computation of convex hulls on general manifolds.

\subsection{Busemann Functions}
  In $\realn d$, the convex hull of a finite set can be described by a finite number of hyperplanes each supported by $d$ points from the set.  A hyperplane through a point may also be thought of as the limiting case of a sphere whose center has been moved away to infinity while a point at its surface remains fixed.
We generalize this notion with the definition of a Busemann function.

For this notion to work, we must restrict ourselves to a class of spaces called $\CAT{0}$ spaces. They are metric spaces with non-positive curvature. Additionally, they must be \emph{complete}; that is, Cauchy sequences in the space must converge to a point in the space. Euclidean space, hyperbolic space, and $\PD{n}$ are all examples of complete $\CAT{0}$ spaces. To talk about ``sending a point away to infinity,'' we must provide a rigorous definition of what we mean by \emph{infinity} in a complete $\CAT{0}$ space.

Two geodesic rays $c_1,c_2:\reals^+\to M$ in a complete $\CAT{0}$ space $M$ are \emph{asymptotic} if $\lim_{t\to\infty}d(c_1(t),c_2(t)) < K$ for some $K\in\reals^+$.  
If $c_1$ and $c_2$ are asymptotic, then we say $c_1 \sim c_2$.  
This forms an equivalence relation $\sim$ so that $[c]$ describes the set of all geodesics $c'$ such that $c \sim c'$.  Let $\xi = [c]$ where $\xi$ is identified with the limit of any geodesic ray asymptotic to $c$.  We say that $\xi$ is a point at infinity. Moreover, for any point $x\in M$ we can find a member of $[c]$ that issues from $x$~\cite[II.8]{Bridson2009}.

\begin{defn}
For a complete $\CAT{0}$ space $M$,  
given a geodesic ray $c(t):\reals^+\to M$, a \emph{Busemann function} $b_c : M \to \reals$ is defined
\[
  b_c(p) = \lim_{t\to\infty}d(p,c(t))-t.
\]
\end{defn}
It should be noted that if we construct a Busemann function from any geodesic ray in $[c]$, it is the same function up to addition by a constant~\cite[II.8]{Bridson2009}. It's convenient then to normalize a Busemann function by assuring that $b_c(I)=0$.

A Busemann function is an example of a \emph{horofunction}~\cite[II.8]{Bridson2009}. A \emph{horosphere} $\hs{h}{r}\subset M$ is a level set of a horofunction $h$; that is, $\hs{h}{r} = h^{-1}(r)$, where $r\in\reals$.  A \emph{horoball} $\hb{h}{r}\subset M$ is a sublevel set of $h$; that is, $\hb{h}{r} = h^{-1}((-\infty,r])$. Horofunctions are convex~\cite[II.8]{Bridson2009}, so any sublevel set of a horofunction is convex, and therefore any horoball is convex.

\vspace{-6 pt}
\paragraph {Example: Busemann functions in $\realn{n}$.} 
As an illustration, we can easily compute the Busemann function in Euclidean space associated with a ray $c(t) = t\mathbf{u}$, where $\mathbf{u}$ is a unit vector. Since $\lim_{t\to\infty}\frac 1{2t}(\|p-t\mathbf{u}\|+t) = 1$, 
\begin{align*}
b_c(p) &
= \lim_{t\to\infty}(\|p-t\mathbf{u}\|-t) \\& 
= \lim_{t\to\infty}\frac 1{2t}(\|p-t\mathbf{u}\|^2-t^2) \\&
= \lim_{t\to\infty}\frac 1{2t}(\|p\|^2-2\inner p{t\mathbf{u}}+\|t\mathbf{u}\|^2-t^2) \\& 
= \lim_{t\to\infty}\frac {\|p\|^2}{2t}-\inner{p}{\mathbf{u}} 
= -\inner{p}{\mathbf{u}}.
\end{align*}

Horospheres in Euclidean space are then just hyperplanes, and horoballs are halfspaces.

\subsubsection{Decomposing $\PD n$}
\label{sec:decomp-pd}

In order to construct Busemann functions in $\PD n$ it is necessary to decompose the space into simpler components. The notion of a \emph{horospherical projection} will be very useful. 

\paragraph{The horospherical group.} 

There is a subgroup of $\GL n$, $N_\xi$ (the \emph{horospherical group}), that leaves the Busemann function $b_c$ invariant~\cite[II.10]{Bridson2009}. That is, given $p\in\PD n$, and $\nu\in N_\xi$, $b_c(\nu p\trn\nu) = b_c(p)$.
Let $A$ be diagonal, 
where $A_{ii}>A_{jj}$, $\forall i > j$. Let $c(t) = e^{tA}$, and $\xi = c(\infty)$.
Then $\nu\in N_\xi$ if and only if $\nu$ is a upper-triangular matrix with 
ones on the diagonal\footnote{For simplicity, we consider only those rays
  with unique diagonal entries, but this definition may be extended to
  those with multiplicity.}.
If $A \in S(n)$ is not sorted-diagonal, we may still use this characterization of $N_\xi$ without loss of generality, since we may compute an appropriate diagonalization $A=QA'\trn Q$, $Q\trn{Q}=I$, then apply the isometry $\trn{Q}pQ$ to any element $p\in\PD n$.

\paragraph{Flats.} 

Let $A\in\Sym n$ and $c(t)=e^{tA}$ as above. If we consider all elements $f\in\PD n$ that share eigenvectors $Q$ with $e^A$,
then all such elements commute with each other and $fe^A=e^Af$.  
We call this space $F$, the \emph{$n$-flat} containing $c$. 
Since we may assume that $Q\in\SO{n}$, every flat $F$ corresponds to an element of $\SO{n}$. Moreover, since members of $F$ commute, $\log(ab) = \log a + \log b $ for all $a,b \in F$.  So if $u$ and $v$ are in $F$, then the distance between them is 
$\symnorm{\log(\inv{u}v)}=\symnorm{(\log(v)-\log(u))}$. 
Since $\symnorm{(\cdot)}$ is a Euclidean norm on $\log(F)$,
we have that $F$ is isometric to $\realn n$ with a Euclidean metric under $\log(\cdot)$.  
  
\paragraph{Horospherical projection.} 

Given $p\in\PD n$, there is a unique decomposition $p=\nu f\trn\nu$ where $(\nu,f)\in N_\xi\times F$~\cite[II.10]{Bridson2009}.
  Let $p\in\PD n$ and $(\nu,f)\in N_\xi\times F$. If $p=\nu f\trn\nu$, then define the \emph{horospherical projection function} $\pi_F:\PD n\to F$ as $\pi_F(p)=\inv\nu p\itr\nu=f$.

\subsubsection{Busemann functions in $\PD n$.}
We can now give an explicit expression for a Busemann function in $\PD n$. For geodesic $c(t)=e^{tA}$, where $A\in\Sym n$, the Busemann function 
$b_c:\PD n\to\reals$ is 
  \[
    b_c(p) = -\tr(A\log(\pi_F(p))),
  \]
  where $\pi_F$ is defined as above~\cite[II.10]{Bridson2009}.

In $\PD 2$ it is convenient to visualize Busemann functions through horospheres.  We can embed $\PD 2$ in $\b{R}^3$ where the log of the determinant of elements grows along one axis.   The orthogonal planes contain a model of hyperbolic space called the \emph{Poincar\'{e} disk} that is modeled as a unit disk, with boundary at infinity represented by the unit circle.  Thus the entire space can be seen as a cylinder, as shown in Figure \ref{fig:horosphere}.  Within each cross section with constant determinant, the horoballs are disks tangent to the boundary at infinity.

\begin{figure}
\begin{center}
\includegraphics[scale=.6]{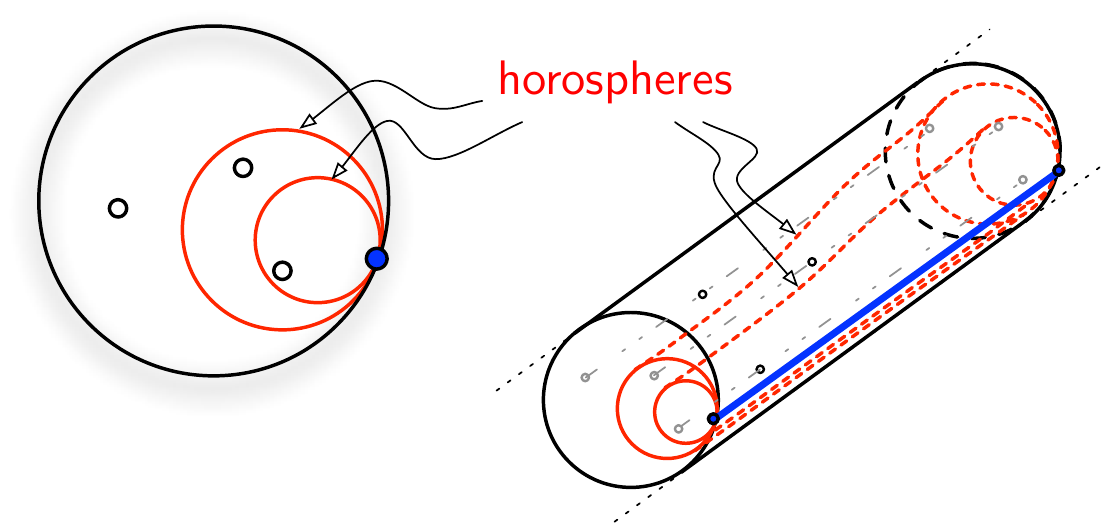}
\end{center}
\vspace{-.3in}
\caption{\label{fig:horosphere}\textsf{\small Left: projection of $X \subset \PD{2}$ onto $\det(x)=1$. Right: $X\subset\PD{2}$.  Two horospheres are drawn in both views.  }}
\vspace{-.15in}
\end{figure}

\section{Ball Hulls}
\label{sec:ball-hulls}

We now introduce our variant of the convex hull in $\PD n$, which we call the ball hull.  
For a subset $X \subset\PD n$, the \emph{ball hull} $\bh X$ is the intersection of all horoballs that also contain $X$: 
        \[
            \bh X = \bigcap_{b_c,r} \hb{b_c}{r},\ X\subset \hb{b_c}{r}.
        \]
Note that the ball hull can be seen as an alternate generalization of the Euclidean convex hull (i.e. via intersection of halfspaces) to $\PD n$.  Furthermore, since it is the intersection of closed sets, it is itself guaranteed to be closed.

\subsection{Properties Of The Ball Hull}
\label{sec:properties-ball-hull}

We know that any horoball is convex. Because the ball hull is the intersection of convex sets, it is itself convex (and therefore $\ch X\subseteq\bh X$).  
We can also show that it shares critical parts of its boundary with the convex hull (Theorem \ref{thm:bdr-pt}), but unfortunately, we cannot represent it as a finite intersection of horoballs (Theorem \ref{thm:infinite-HB}).

    \begin{theorem}
       Every $x\in X$ ($X$ finite) on the boundary of $\bh X$ is also on the boundary of $\ch X$ (i.e., $X\cap\bbh X\subseteq X\cap\bch X$).
\label{thm:bdr-pt}
    \end{theorem}
    \begin{proof}
        Since $X\subset\ch X$, either $x\in\bch X$ or $x\in\ich X$. Assume that $x\in\ich X$. Then there is a neighborhood $U$ of $x$ contained wholly in $\ch X$. Because $x\in\bbh X$, there is a horofunction $h$ such that $X\subset\hb h r$ and $h(x) = r$.  Since $\hb h r$ is convex, $U\subset\ch X\subseteq\hb h r$. This implies that $h(x)<r$, but $h(x) = r$, a contradiction. Thus $x\in\bch X$.
    \end{proof}

\begin{theorem}
In general, the ball hull cannot be described as the intersection of a finite set of horoballs.
\label{thm:infinite-HB}
\end{theorem}
\begin{proof}
We construct an example in $\PD{4}$ with a point set $X = \{x_1, x_2\}$ of size $2$ where the ball hull cannot be described as the intersection of a finite number of horoballs.  Let $X \subset \hboln{3}$, the three dimensional hyperbolic space, as embedded in $\PD{4}$.
In particular, let the geodesic that contains $x_1$ and $x_2$ also contain $I$, the identity matrix, at their midpoint on the geodesic.

Consider the family of horofunctions $\Eu{H}$ such that for $h \in \Eu{H}$, $h(x_1) = h(x_2)$. By construction, $h(x_1)=h(x_2)=r$ for some $r\in\reals$. In the Poincar\'e ball model of $\hboln{3}$, the horospheres $\hs{h}{r}$ are literally spheres that are tangent to the boundary at infinity, and touch $x_1$ and $x_2$. So constructed, any horosphere in $\Eu{H}$ will contact the boundary at infinity on a great circle that is equidistant from both points.

The ball hull $\bh X$ is defined $\{\bigcap \hb{h}{r} \mid h \in \Eu{H}\}$, and is a ``spindle'' (a three dimensional lune) with tips at $x_1$ and $x_2$ and bulges out from the geodesic segment between them.  Any finite family of horofunctions will intersect in a region strictly larger than $\bh X$, so every horoball generated by a member in $\Eu{H}$ is necessary for $\bh X$, and so there is no finite set of horoballs that describe $\bh X$.
\end{proof}

\section{The $\eps$-Ball Hull}
\label{sec:algorithm-eps-ball}

Theorem \ref{thm:infinite-HB} indicates that we cannot maintain a finite representation of a ball hull.  However, as we shall show in this section, we can maintain a finite-sized \emph{approximation} to the ball hull. Our approximation will be in terms of \emph{extents}: intuitively, we say that a set of horoballs approximates the ball hull if a geodesic traveling in any direction traverses approximately the same distance inside the ball hull as it does inside the approximate hull. 

\parpic[r]{\PPicCap{2cm}{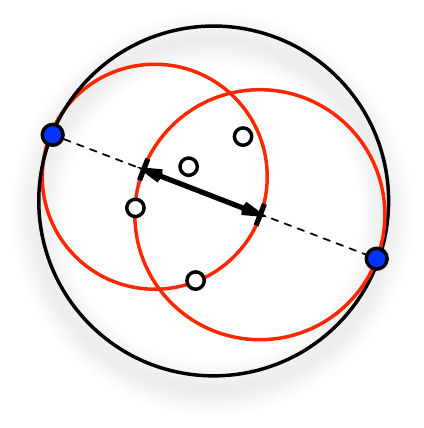}{\hspace{.1in} \textsf{\small horoextent}}}
In Euclidean space, we can capture extent by measuring the distance between two parallel hyperplanes that sandwich the set. Since measuring extent by recording the distance between two parallel planes does not have a direct analogue in $\PD n$, we define a notion we call a \emph{horoextent}.
  Let $c(t)=qe^{t\inv{q}A}$ be a geodesic, and $X\subset\PD n$. The \emph{horoextent} $\ext{c}{X}$ with respect to $c$ is defined as:  
\begin{equation*}
    \ext{c}{X} 
    = \left| \max_{p\in X}b_{c+}(p) + \max_{p\in X}b_{c-}(p) \right|,
\end{equation*}
  where $b_{c+}$ is the Busemann function created when we allow $t$ to approach positive infinity as normal, while $b_{c-}$ is the Busemann function created when we allow the limit to go the other direction; that is: 
  \[
    b_{c+}(p) = \lim_{t\to+\infty}(d(c(t),p)-t),\quad 
    b_{c-}(p) = \lim_{t\to-\infty}(d(c(t),p)+t).
  \] 
  Observe that for any $c$, $\ext{c}{X} = \ext{c}{\ch X} = \ext{c}{\bh X}$.
  (Note that we cannot simply substitute $\min b_{c+}$ for $\max b_{c-}$; the horoballs are generated by Busemann functions tied to \emph{opposite} points at infinity.)
  
  If we were to compute $\ext{c}{X}$ in a Euclidean space, it would be apparent that the extent would be the distance between two parallel planes. 
  Because the minimum distance between two Euclidean horoballs is a constant, no matter where we place the base point $q$, we can measure horoextent simply by measuring width in a particular direction. 
  However, this is not true in general. For instance in $\PD{n}$, horofunctions are nonlinear, so the distance between opposing horoballs is not constant. The width of the intersection of the opposing horoballs is taken along the geodesic $c$, and a geodesic is described by a point $q$ and a direction $A$. 
  We fix the point $q$ so that we need only choose a uniform grid of directions $A$ for our approximation.

\begin{defn}
\label{defn:eps-ball-hull}
  An intersection of horoballs is called an \emph{$\eps$-ball hull with origin $q$} ($\ebh X$)  if for all geodesic rays $c$ such that $c(0)=q$, $|\ext{c}{\ebh X} - \ext{c}{X}|\leq\eps$.  When $q$ is clear, we will refer to $\ebh X$ as just an $\eps$-ball hull. 
\end{defn}

\paragraph{Shifting the origin.}  
Let the \emph{geodesic anisotropy}~\cite{M-B-} of a point $p\in\PD n$ be defined as $\GA(p)=d(\sqrt[n]{\det(p)}I, p)$ (so if $\det(p) = 1$ then $\GA(p) = d(I,p)$).  Let $d_X = \max_{p \in X} d(p,I) \geq \max_{p \in X} \GA(p)$.  
The size of the $\eps$-ball hulls we construct will depend $d_X$, but this is not an intrinsic parameter of the data, since we can change it merely by isometrically translating the point set.  
For some point $q \in \ch{X}$, if we could translate the data set so that $q$ was at the origin $I$, then $d_X \leq \diam{X} = \max_{p,q \in X} d(p,q)$.  
We now prove that such a translation is always possible. 

\begin{lemma}
\label{lem:shifting-origin}
  For a point $q\in\PD{n}$, a geodesic $c$ such that $c(0)=q$, and a point set $X$, if we define an operation $\hat{S}$ on a set $S\in\PD{n}$ such that $\hat{p}=q^{-\frac{1}{2}}pq^{-\frac{1}{2}}$ for any $p\in S$, then 
  \[
    \ext{c}{X} = \ext{\hat{c}}{\hat{X}}.
  \]
\end{lemma}
\begin{proof}

  Let $c(t)=qe^{t\inv{q}A}$. Then 
  $
    \hat{c}(t) 
    = q^{-\frac{1}{2}}(qe^{t\inv{q}A})q^{-\frac{1}{2}}
    = q^{\frac{1}{2}}e^{t\inv{q}A}q^{-\frac{1}{2}}
    = e^{tq^{-\frac{1}{2}}Aq^{-\frac{1}{2}}}
    = e^{t\hat{A}}.
  $
  Note that 
  $
    \|A\|_q 
    = \sqrt{\tr((q^{-\frac{1}{2}}Aq^{-\frac{1}{2}})^2)}
    = \|\hat{A}\|_I
  $, so $\hat{c}$ is a geodesic such that $\hat{c}(0)=I$ with the same speed as $c$. 
  Let $b_c$ be the Busemann function of $c$, so $b_c(p)=\lim_{t\to\infty}(d(c(t),p)-t)$. Since conjugation by $q^{-\frac{1}{2}}$ is an isometry on $\PD{n}$, 
  \[
    b_c(p)
    = \lim_{t\to\infty}(d(c(t),p)-t)
    = \lim_{t\to\infty}(d(\hat{c}(t),\hat{p})-t)
    = b_{\hat{c}}(\hat{p}),
  \]
  and therefore
  \[
    \ext{c}{X} 
    = \ext{\hat{c}}{\hat{X}}.
  \]
\end{proof}

For convenience, we will now assume that our data has been shifted into a reasonable frame where some point $q \in \ch{X}$ is the base point of our horofunction.  That is in this shifted frame $I \in \ch{X}$ and, we can bound, $d_X \leq \diam{X}$.

\paragraph{Main result.}
The main result of this section is a construction of a finite-sized $\eps$-ball hull.

\begin{theorem}
 For a set $X \subset \PD n$ of size $N$ (for constant $n$), we can construct an $\eps$-ball hull of size $O((\sinh(d_X)/\eps)^{n-1}\cdot N^{\lfloor n/2 \rfloor})$ in time $O((\sinh(d_X)/\eps)^{n-1} (N^{\lfloor n/2 \rfloor} + N \log N))$. 
 \label{thm:epsBH-alg}
\end{theorem}

\paragraph{Proof Overview.}
\label{sec:proof-overview}

We make much use of the structure of flats in our proof, so it is helpful to describe some conventions. Consider the set of unit-length tangent vectors at $I$, part of the tangent space $\Sym{n}$; in other words, the set of ``directions'' from $I$. If we choose a flat $F$ to work in, then the tangent space of $F$ contains a subset of those directions. All these directions, though, share the rotation $Q$ identified with $F$. So in much of the rest of the paper, we refer to this rotation $Q$ as a ``direction,'' even though it is not a member of the tangent space. 

Our proof uses two key ideas. First, we show that within a flat $F$ (i.e., given a direction $Q \in \SO n$) we can find a finite set of minimal horoballs exactly.  This is done by showing an equivalence between halfspaces in $F$ and horoballs in $\PD n$ in Section \ref{sec:projection-k-flat}.  The result implies that computing minimal horoballs with respect to a direction $Q$ is equivalent to computing a convex hull in Euclidean space.

Second, we show that instead of searching over the entire space of directions $\SO n$, we can discretize it into a finite set of directions such that when we calculate the horoballs with respect to each of these directions, the horoextents of the resulting $\eps$-ball hull are not too far from those of the ball hull.  In order to do this, we prove a Lipschitz bound for horofunctions (and hence horoextents) on the space of directions.  
Since any two flats $F$ and $F'$ are identified with rotations $Q$ and $Q'$, we can move a point from $F$ to $F'$ simply by applying the rotation $\trn{Q}Q'$, and measure the angle $\theta$ between the flats.  If we consider a geodesic $c\subset F$ such that $c(0) = I$, we can apply $\trn{Q}Q'$ to $c$ to get $c'$, then for any point $p \in \PD n$ we bound $|b_{c}(p) - b_{c'}(p)|$ as a function of $\theta$.  

Proving this theorem is quite technical. We first prove a Lipschitz bound in $\PD 2$, where the space of directions is a circle (as in the left part of Figure \ref{fig:horosphere}). After providing a bound in $\PD 2$ we decompose the distance between two directions in  $\SO n$ into ${n \choose 2}$ angles defined by $2 \times 2$ submatrices in an $n \times n$ matrix.  In this setting it is possible to apply the $\PD 2$ Lipschitz bound ${ n \choose 2}$ times to get the full bound.  The proof for $\PD 2$ is presented in Section~\ref{sec:lipschitz-bound-pd2}, and the generalization to $\PD n$ is presented in Section~\ref{sec:generalizing-pd-n}. Finally, we combine these results in an algorithm in Section~\ref{sec:slow-algorithm}.

The following lemma describes how geodesics (and horofunctions) are transformed by a rotation. 

\begin{lemma}
\label{lem:horo-rotate}
  For a point $p\in\PD{n}$, a rotation matrix $Q$, geodesics $c(t)=e^{tA}$ and $c'(t)=e^{t\actx{Q}{A}}$, 
  \[
    b_{c'}(p) = b_c(\actxt{Q}{p}).
  \]
\end{lemma}
\begin{proof}
  If $A'=\actx{Q}{A}$ is the tangent vector of $c'$, and $F'$ is the flat containing $c'$,
  \begin{align*}
    b_{c'}(p) 
    &= -\tr(A'\log(\pi_{F'}(p)))
    = -\tr((\actx{Q}{A})\log(\actx{(\actx{Q}{\inv{\nu}})}{p}))\\
    &= -\tr(A\actxt{Q}{\log(\actx{Q}{\actxi{\nu}{(\actxt{Q}{p})}})})
    = -\tr(A\log(\actxi{\nu}{(\actxt{Q}{p})})) \\
    &= b_c(\actxt{Q}{p}).
  \end{align*}
\end{proof}

In particular, this allows us to pick a flat where computation of $b_c$ is convenient, and rotate the point set by $Q$ to compute $b_c$ instead of attempting computation of $b_{c'}$ directly, which may be more cumbersome; we will utilize this idea later.

\subsection{Projection to $k$-flat}
\label{sec:projection-k-flat}

For the first part of our proof, we establish an equivalence between horospheres and halfspaces. That is, after we compute the projection of our point set, we can say that the point set lies inside a horoball $\hb{b_c}{r}$ if and only if its projection lies inside a halfspace $H_r$ of $F$ (recall that $F$ is isometric to a Euclidean space under $\log$).

\begin{lemma}   \label{lem:horo-hyper-dual}
  For any horoball $\hb{b_c}{r}$, there is a halfspace $H_r\subset\log(F)\subset\Sym{n}$ such that $\log(\pi_F(\hb{b_c}{r})) = H_r$.
\end{lemma}
\begin{proof}
  If $b_c(p)\leq r$, $p\in\PD n$, and $c(t)=e^{tA}$, then $-\tr(A\log(\pi_F(p)))\leq r$. Since $\pi_F(p)$ is positive-definite, $\log(\pi_F(p))$ is symmetric. But $\tr((\cdot)(\cdot))$ defines an inner product on the Euclidean space of symmetric $n\times n$ matrices. Then the set of all $Y$ such that $-\tr(AY)\leq r$ defines a halfspace whose boundary is perpendicular to $A$.
\end{proof}

This gives us a means to compute horoballs by using $\pi_F$ to project our point set onto $F$, and leverage a familiar Euclidean environment.

\subsection{A Lipschitz bound in $\PD 2$}
\label{sec:lipschitz-bound-pd2}

\subsubsection{Rotations in $\PD 2$}
\label{sec:rotations-pd-2}

We start with some technical lemmas that describe the locus of rotating points in $\PD 2$.

\begin{lemma}
\label{lem:rotate}
  Given a rotation matrix $Q\in\SO{2}$ corresponding to an angle of $\theta/2$, $Q$ acts on a point $p\in\PD{2}$ via $\actx{Q}{p}$ as a rotation by $\theta$ about the (geodesic) axis $e^{tI}=e^tI$.
\end{lemma}
\begin{proof}
  If $p = e^t I$, $t\in\reals$, then 
  $
    \actx{Q}{p}=e^t\actx{Q}{I}=e^t I,
  $ 
  so $e^t I$ is invariant under the action of $Q$. Any action $\actx{G}{p}$ where $G\in\GL{n}$ is an isometry on $\PD{n}$, so the distance from $p$ to the axis $e^t I$ remains fixed~\cite[II.10]{Bridson2009}. 
  Computing $\actx{Q}{p}$ as a function of $\theta$, we get:
  \[
    \begin{pmatrix}
      \cos\frac{\theta}{2} & -\sin\frac{\theta}{2} \\
      \sin\frac{\theta}{2} & \cos\frac{\theta}{2}
    \end{pmatrix}
    \mattwosym{u}{v}{w}
    \begin{pmatrix}
      \cos\frac{\theta}{2} & \sin\frac{\theta}{2} \\
      -\sin\frac{\theta}{2} & \cos\frac{\theta}{2} 
    \end{pmatrix}
    = 
    \begin{pmatrix}
      \frac{u+v}{2} + \frac{u-v}{2}\cos\theta - w\sin\theta & 
      \frac{u-v}{2}\sin\theta + w\cos\theta \\
      \frac{u-v}{2}\sin\theta + w\cos\theta & 
      \frac{u+v}{2} - \frac{u-v}{2}\cos\theta + w\sin\theta \\
    \end{pmatrix},
  \]
  which is $2\pi$-periodic. (In fact, it is easy to see a rotation in the ``coordinates" $\frac{u-v}{2}$ and $w$.)
\end{proof}

By Lemma~\ref{lem:rotate}, we know that as we apply a rotation to $p$, it moves in a circle. Because any rotation $Q$ has determinant $1$, $\det(\actx{Q}{p})=\det(p)$. This leads to the following corollary:

\begin{corollary}
In $\PD 2$, the radius of the circle that $p$ travels on is $\GA(p)$
. 
Such a circle lies entirely within a submanifold of constant determinant. 
\label{cor:circle-const-det}
\end{corollary}

In fact, any submanifold $\PD{2}_r$ of points with determinant equal to some $r\in\reals^+$ is isometric to any other such submanifold $\PD{2}_s$ for $s\in\reals^+$. This is seen very easily by considering the distance function $\tr(\log(\inv{p}q))$ --- the determinants of $p$ and $q$ will cancel. 

We pick a natural representative of these submanifolds, $\PD{2}_1$. This submanifold forms a complete metric space of its own that has special structure:

\begin{lemma}
\label{lem:constant-curve}
  $\PD{2}_1$ has constant sectional curvature $-\frac{1}{2}$.
\end{lemma}
\begin{proof}
  Let $p\in\PD{2}_1$. Then if $p=\mattwosym{x}{y}{w}$, $\det(p) = xy-w^2 = 1$. Let $u=\frac{x+y}{2}$ and $v=\frac{x-y}{2}$ so that $x=u+v$ and $y=u-v$. Then $\det(p) = u^2-v^2-w^2=1$. This describes a hyperboloid of two sheets, and restricting $u>0$, is a model for hyperbolic space $\hboln{2}$. 
  
  To analyze the metrics between the two spaces, we may consider our other point to be the identity matrix, since in $\PD{2}_1$, $d(p,q)=d(q^{-1/2}pq^{-1/2},I)$. The equivalent point on the hyperboloid to $I$ is $(u,v,w)=(1,0,0)$. The distance between the two points in the hyperbolic metric is 
  \[
    d_{\hboln{2}}(p,I)
    = \inv{\cosh}(u_1u_2-v_1v_2-w_1w_2)
    = \inv{\cosh}\left(\frac{x+y}{2}\right)
    = \ln \left(\frac{x+y}{2}+\sqrt{\left(\frac{x+y}{2}\right)^{2}-1} \right).
  \]
  And in the metric of $\PD{n}$:
  \[
    d_{\PD{2}}(p,I)
    = \sqrt{\tr(\log^2(p))}
    = \sqrt{\ln^2\la_1+\ln^2\la_2}
    = \sqrt{2}\ln\la_1
    = \sqrt{2}\ln\left(\frac{x+y}{2}+\sqrt{\left(\frac{x+y}{2}\right)^{2}-1} \right).
  \]
  Since this is a constant multiple of the hyperbolic metric, $\PD{2}_1$ is a complete metric space of constant negative sectional curvature. We can find the curvature $\ka$ by solving $1/\sqrt{-\ka}=\sqrt{2}$ to get $\ka = -1/2$~\cite[I.2]{Bridson2009}.
\end{proof}

\subsubsection{Bounding $\|\nabla b_c\|$}
\label{sec:bounding-nabla-b_c}

To bound the error incurred by discretizing the space of directions, we need to understand the behavior of $b_c$ as a function of a rotation $Q$. We show that the derivative of a geodesic is constant on $\PD n$.

\begin{lemma}
  For a geodesic ray $c(t)=e^{tA}$, $\|A\|=1$, then $\|\nabla b_c\| = 1$ at any point $p\in\PD{n}$.
  \label{lem:nabla}
\end{lemma}
\begin{proof}
  Let $u=b_c(p)$. If we define the map $\gamma_p$ to map $p$ to its projection onto any horoball $\hb{b_c}{u-t}$, where $t>0$, then $\gamma_p$ is a geodesic ray~\cite[II.8]{Bridson2009}. Since any $\gamma_p(t)$ is the projection onto the horoball $\hb{b_c}{u-t}$, the geodesic segment $[p,\gamma_p(t)]$ is perpendicular to $\hb{b_c}{u-t}$ at $\gamma_p(t)$, and therefore the tangent vector of $\gamma_p$ points directly opposite to $\nabla b_c$ at $\gamma_p(t)$.
  Because $\gamma_p$ intersects $\hb{b_c}{u-t}$ at $t$, $b_c(\gamma_p(t))$ must change at a rate opposite to $\gamma_p(t)$, along $\gamma_p$, and since $\nabla b_c$ points in the opposite direction as $\gamma_p'(t)$ at $\gamma_p(t)$, $\nabla b_c = -\gamma_p'(t)$.
  
  Also, since $d(p,\hb{b_c}{u-t})=t$ for any $p\in\hb{b_c}{u}$, for any $u$, and $\|\gamma_p'\|$ is constant along $\gamma_p$, $\|\nabla b_c\|$ is the same anywhere in $\PD{n}$.  Since $c(t)$ is the projection of $c(0)$ onto $\hb{b_c}{-t}$ by construction, and geodesics are unique in $\PD{n}$, $\|\nabla b_c\|=\| A \| = 1$.

\end{proof}

\subsubsection{A Lipschitz condition on Busemann functions in $\PD 2$}
\label{sec:lipsch-cond-busem}

We are now ready to prove the main Lipschitz result in $\PD 2$. We start with a more specific bound that depends on the geodesic anisotropy of a point:

\begin{lemma}
\label{lem:Lipschitz}
  Given a point $p\in\PD 2$, a rotation matrix $Q$ corresponding to an angle of $\theta/2$, geodesics $c(t) = e^{tA}$ and $c'(t) = e^{tQ A\trn Q}$, 
  \[
    |b_c(p) - b_{c'}(p)| 
    \leq |\theta| \cdot \sqrt{2}\sinh\left(\frac{\GA(p)}{\sqrt{2}}\right).
  \]
\end{lemma}
\begin{proof}
  The derivative of a function $f$ along a curve $\gamma(t)$ has the form $\inner{\nabla f|_{\gamma(t)}}{\gamma'(t)}$, and has greatest magnitude when the tangent vector $\gamma'(t)$ to the curve and the gradient $\nabla f|_{\gamma(t)}$ are parallel. When this happens, the derivative reaches its maximum at $\|\nabla f|_{\gamma(t)}\|\cdot\|\gamma'(t)\|$. 
  
  Since $\|\nabla b_c\|=1$ anywhere by Lemma~\ref{lem:nabla}, the derivative of $b_c$ along $\gamma$ at $\gamma(t)$ is bounded by $\|\gamma'(t)\|$. We are interested in the case where $\gamma(\theta)$ is the circle in $\PD{2}$ defined by tracing $\actx{Q(\theta/2)}{p}$ for all $-\pi<\theta\leq\pi$. By Corollary~\ref{cor:circle-const-det}, we know that this circle has radius $\GA(p)$ and lies entirely within a submanifold of constant determinant, which by Lemma~\ref{lem:constant-curve} also has constant curvature $\kappa = -1/2$.

This implies that \[
    \|\gamma'(\theta)\|
    = \frac{1}{\sqrt{-\kappa}}\sinh(\sqrt{-\kappa}\,r) 
    = \sqrt{2}\sinh\left(\frac{\GA(p)}{\sqrt{2}}\right).
  \]
  for any value of $\theta\in(-\pi,\pi]$~\cite[I.6]{Bridson2009}.
  Then \[
    |b_c(p) - b_{c'}(p)| 
    = |b_{c}(p) - b_{c}(\actxt{Q}{p})| 
    \leq |\theta| \cdot \sqrt{2}\sinh\left(\frac{\GA(p)}{\sqrt{2}}\right).
  \]
\end{proof}

We can now state our main Lipschitz result in $\PD 2$.

\begin{theorem}
\label{thm:diam-bound}
For $Q\in\SO{2}$ corresponding to $\theta/2$, and let $\gamma(t)=e^{tA}$ and $\gamma'(t)=e^{t\actxt{Q}{A}}$. Then for any $p \in X$
  \[
    |b_\gamma(p) - b_{\gamma'}(p)| 
    \leq |\theta| \cdot \sqrt{2}\sinh\left(\frac{d_X}{\sqrt{2}}\right).
  \]
\end{theorem}

\subsection{Generalizing to $\PD n$}
\label{sec:generalizing-pd-n}

  Now to generalize to $\PD n$ we need to decompose the projection operation $\pi_F(\cdot)$ and the rotation matrix $Q$.  We can compute $\pi_F$ recursively, and it turns out that this fact helps us to break down the analysis of rotations. Since we can decompose any rotation into a series of $2 \times 2$ rotation matrices, decomposing the computation of $\pi_F$ in a similar manner lets us build a Lipschitz condition for $\PD n$.

\begin{lemma}
  
  Let $F\subset\PD{n}$ be the flat containing diagonal matrices, let $r+s=n$, and let $\pi_{F,r}$ and $\pi_{F,s}$ be the projection operations for $r\times r$ and $s\times s$ flats of diagonal matrices, respectively.  Then 
  \[
    \pi_F(p) = \mattwosym{\pi_{F,r}(h_r)}{\pi_{F,s}(p_s)}{},
  \]
  where $p_s$ is the lower right $s\times s$ block of $p$, and $h_r$ is the Schur complement of $p_s$.
\end{lemma}
\begin{proof}
  $\pi_F$ is computed by decomposing $p$ into $\actxi{\nu}{f}$, where $f$ is diagonal (and positive-definite) and $\nu$ is upper-triangular, with ones on the diagonal. This can be done by computing the Schur complement of some lower right square block of $p$, and putting the complement in the upper right block of a new matrix, with the lower block in the other corner. The original matrix can be reconstructed by conjugating by an upper-triangular matrix of appropriate form:
  \[
    p=\mattwobsym{p_r}{p_s}{a} = 
    \begin{pmatrix}
      I_r & a\inv{p_s} \\
      & I_s
    \end{pmatrix}
    \mattwosym{h_r}{p_s}{}
    \begin{pmatrix}
      I_r & \\
      \inv{p_s}\trn{a} & I_s
    \end{pmatrix}.
  \]
  
  Performing this process recursively yields a product of upper-triangular matrices with ones on the diagonal, which is again an upper-triangular matrix with ones on the diagonal, yielding $\nu$, and a diagonal matrix, $f=\pi_F(p)$.
\end{proof}

If we wish to compute $\pi_F(p)$ for other flats, then we can apply the rotation of $F$ to $p$, compute using the recursive formula described above, and then apply the opposite rotation to the resulting diagonal matrix. Most of the time, however, it is most convenient to condition our data so that $\pi_F$ is computed for a diagonal flat $F$.

We now wish to analyze a simpler form of rotation, one that can be broken into rotations on separate axes.

\begin{lemma}
  \label{lem:Rotation-dc}
  Given a point $p\in\PD n$, a rotation matrix 
  $
    Q=\mattwosym{Q_r}{Q_s}{}, 
  $ 
  such that $r+s=n$ and $Q_r$, $Q_s$ are $r\times r$, $s\times s$ rotation matrices, respectively, and a geodesic $c(t) = e^{tQ A\trn Q}$ with $A=\mattwosym{A_r}{A_s}{}$ sorted-diagonal, 
  \[
    b_c(p) 
    = -\tr(A_r\log(\pi_{F,r}(\actxt{Q_r}{h_r})))-\tr(A_s\log(\pi_{F,s}(\actxt{Q_s}{p_s}))).
  \]
\end{lemma}
\begin{proof}
  From Lemma \ref{lem:horo-rotate} we know that $b_c(p) = -\tr(A\log(\pi_F(\actxt{Q}{p})))$.
  Compute $\pi_F(\actxt{Q}{p})$ by first decomposing $p\mapsto\mattwobsym{p_r}{p_s}{a}$:
  \[
    \actxt{Q}{p} 
    = \mattwosym{\trn{Q_r}}{\trn{Q_s}}{}\mattwobsym{p_r}{p_s}{a}\mattwosym{Q_r}{Q_s}{}
    = \begin{pmatrix}
      \actxt{Q_r}{p_r} & \trn{Q_r}a Q_s \\
      \trn{Q_s}\trn{a} Q_r & \actxt{Q_s}{p_s}
    \end{pmatrix}.
  \]
  Now compute the Schur complement of $\actxt{Q_s}{p_s}$:
  \begin{align*}
    \actxt{Q_r}{p_r}-\trn{Q_r}a Q_s\inv{(\actxt{Q_s}{p_s})}\trn{Q_s}\trn{a}Q_r
    &= \actxt{Q_r}{p_r}-\trn{Q_r}a Q_s\actxt{Q_s}{\inv{p_s}}\trn{Q_s}\trn{a}Q_r\\
    &= \actxt{Q_r}{p_r}-\trn{Q_r}\actx{a}{\inv{p_s}}Q_r\\
    &= \actxt{Q_r}{(p_r-\actx{a}{\inv{p_s}})}.
  \end{align*}
  But $h_r = p_r-\actx{a}{\inv{p_s}}$ is just the Schur complement of $p_s$, so 
  \begin{align*}
    b_c(p) = -\tr(A\log(\pi_F(\actxt{Q}{p}))) &
    = -\tr\mattwosym{A_r\log(\pi_{F,r}(\actxt{Q_r}{h_r}))}{A_s\log(\pi_{F,s}(\actxt{Q_s}{p_s}))}{} \\&
    = -\tr(A_r\log(\pi_{F,r}(\actxt{Q_r}{h_r}))) -\tr(A_s\log(\pi_{F,s}(\actxt{Q_s}{p_s}))).
  \end{align*}
\end{proof}

This allows us to break the Lipschitz bound into smaller pieces that we can analyze individually. The following two corollaries give us a way to analyze the effects of $2\times 2$ rotation matrices:

\begin{corollary}
\label{cor:adjacent-rot}
  Given a point $p\in\PD n$, a rotation matrix 
  $
    Q=\begin{pmatrix}
      I_r & & \\
      & Q' & \\
      & & I_s 
    \end{pmatrix}, 
  $ 
  where $r+s+2=n$, $Q'$ is a $2\times 2$ rotation matrix corresponding to an angle of $\theta/2$, geodesics $c(t) = e^{tA}$ and $c'(t) = e^{tQ A\trn Q}$, then $|b_c(p) - b_{c'}(p)|$ is bounded as in Lemma \ref{lem:Lipschitz}. 
\end{corollary}
\begin{proof}
  This is easily seen after observing that $I_r$ and $I_s$ are also rotation matrices, so Lemma \ref{lem:Rotation-dc} can be applied twice.
\end{proof}

\begin{corollary}
  The results of Corollary \ref{cor:adjacent-rot} extend to rotations between any two coordinates, that is, where $Q'$ is of the form
  \[
    \begin{pmatrix}
      \cos(\theta/2) & & -\sin(\theta/2) \\
      & I_t & \\
      \sin(\theta/2) & & \cos(\theta/2) \\
    \end{pmatrix}.
  \]
\end{corollary}
\begin{proof}
  First observe that a rotation matrix $Q_{i,j}$ that rotates between axes $i$ and $j$ is equal to a matrix $E_{i+1,j}Q_{i,i+1}\trn{E_{i+1,j}}$, where $E_{i,j}$ is a permutation that moves row $i$ to row $j$ and shifts the intervening rows up. We assume $E=E_{i+1,j}$, $Q'=Q_{i,i+1}$, and $Q=\actx{E}{Q'}$ from here on.
  
  Assuming that $A$ is sorted-diagonal, we can compute $b_{c'}(p)$ as:
  \begin{align*}
    b_{c'}(p) 
    &= -\tr(\actx{(\actx{E}{Q'})}{A}\log(\actx{(\actx{(\actx{E}{Q'})}{\inv{\nu}})}{p})) \\
    &= -\tr((\actxt{E}{A})\log(\actx{(\actxt{E}{\inv{\nu}})}{\actxt{Q'}{(\actxt{E}{p})}})) \\
    &= -\tr(\hat A\log(\actxi{\hat\nu}{\actxt{Q'}{\hat p}})) \\
    &= -\tr(\hat A\log(\pi_{\hat F}(\actxt{Q'}{\hat p}))), 
  \end{align*}
  which can be computed as above; some care must be taken, however, since the order of elements of $\hat A$ is different than that of $A$. That is, in certain places, the Schur complement of the \emph{upper} corner must be taken to compute $\pi_{\hat F}$, rather than that of the lower corner.
\end{proof}

\subsubsection{A Lipschitz condition on Busemann functions in $\PD n$}
\label{sec:lipsch-PDn}
We are now ready to prove the main Lipschitz result in $\PD n$. We start with a more specific bound that depends on the distance from a point $p$ to $I$:

\begin{lemma}
\label{lem:Lipschitz-PDn}
  Given a point $p\in\PD n$, a rotation matrix $Q \in \SO{n}$ corresponding to an angle of $\theta/2$, geodesics $c(t) = e^{tA}$ and $c'(t) = e^{tQ A\trn Q}$, 
  \[
    |b_c(p) - b_{c'}(p)| 
    \leq |\theta| \cdot {n \choose 2} \cdot \sqrt{2}\sinh\left(\frac{d(p,I)}{\sqrt{2}}\right).
  \]
\end{lemma}

\begin{proof}
  Every rotation $Q$ may be decomposed into a product of rotations $Q=Q_1 Q_2\dots Q_k$ where $k={n \choose 2}$ and $Q_i$ is a $2\times 2$ sub-block rotation corresponding to an angle of $\theta_i/2$ with $|\theta_i| \leq |\theta|$.  
  Then 
  \[
    |b_c(p) - b_{c'}(p)| 
    = \left|\sum_{i=1}^k (b_{c'}^{i-1}(p) - b_{c'}^i(p))\right|
    \leq \sum_{i=1}^k |b_{c'}^{i-1}(p) - b_{c'}^i(p)|,
  \]
  where $b_{c'}^0(p) = b_c(p)$ and $b_{c'}^i(p)$ is $b_c(p)$ with the first $i$ rotations successively applied, so if $Q_i'=\prod_{j=1}^i Q_j$,
  \[
    b_{c'}^i(p) = b_c(\actxt{(Q_i')}{p}).
  \]
  But then 
  \[
    |b_{c'}^{i-1}(p) - b_{c'}^i(p)| 
    \leq |\theta_i| \cdot \sqrt{2}\sinh\left(\frac{d(p,I)}{\sqrt{2}}\right),
  \]
  and therefore 
  \[
    |b_c(p) - b_{c'}(p)|
    \leq \left(\sum_{i=1}^k|\theta_i|\right) \cdot \sqrt{2}\sinh\left(\frac{d(p,I)}{\sqrt{2}}\right)
    \leq
    |\theta| \cdot {n \choose 2} \cdot \sqrt{2}\sinh\left(\frac{d(p,I)}{\sqrt{2}}\right),
  \]
since for all $i$ we have $|\theta_i| \leq |\theta|$.
\end{proof}

We can now state our main Lipschitz result in $\PD n$.  

\begin{theorem}[Lipschitz condition on Busemann functions in $\PD n$]
\label{lem:Lipschitz-gen} 
Consider a set $X \subset \PD n$,
a rotation matrix $Q\in\SO{n}$ corresponding to an angle $\theta/2$, geodesics $c(t) = e^{tA}$ and $c'(t) = e^{tQ A\trn Q}$.  
Then for any $p \in X$
  \[
    |b_c(p) - b_{c'}(p)| 
    \leq |\theta| \cdot {n \choose 2} \cdot \sqrt{2}\sinh\left(\frac{d_X}{\sqrt{2}}\right).
  \]
\end{theorem}

\subsection{Algorithm}
\label{sec:slow-algorithm}
For $X \subset \PD n$ we can construct $\eps$-ball hull as follows.  
We place a grid $G_\eps$ on $\SO n$ so that for any $Q' \in \SO n$, there is another $Q \in G_\eps$ such that 
the angle between $Q$ and $Q'$ is at most 
$(\eps/2)/(2{n \choose 2} \sqrt{2}\sinh(d_{X}/{\sqrt{2}}))$.  
For each $Q \in G_\eps$, we consider $\pi_F(X)$, the projection of $X$ into the associated $n$-flat $F$ associated with $Q$.  Within $F$, we construct a convex hull of $\pi_F(X)$, and return the horoball associated with each hyperplane passing through each facet of the convex hull, as in Lemma \ref{lem:horo-hyper-dual}.  

To analyze this algorithm we can now consider any direction $Q' \in \SO n$ and a horofunction $b_{c'}$ that lies in the associated flat $F'$.  There must be another direction $Q \in G_\eps$ such that the angle between $Q$ and $Q'$ is at most $(\eps/2)/(2{n \choose 2} \sqrt{2}\sinh(d_{X}/\sqrt{2}))$.  Let $b_{c}$ be the similar horofunction to $b_{c'}$, except it lies in the flat $F$ associated with $Q$.  This ensures that for any point $p \in X$, we have $|b_{c'}(p) - b_{c}(p)| \leq \eps/2$.  
Since $\ext{c'}{X}$ depends on two points in $X$, and each point changes at most $\eps/2$ from $b_{c'}$ to $b_{c}$ we can argue that $|\ext{c'}{X} - \ext{c}{X}| \leq \eps$. 
Since this holds for any direction $Q' \in \SO n$ and for $Q \in G_\eps$ the function $\ext{c}{X}$ is exact, the returned set of horoballs defines an $\eps$-ball hull.  

Since (with constant $n$) the grid $G_\eps$ is of size $O((\sinh(d_X)/\eps)^{n-1})$ and computing the convex hull in each flat takes $O(N^{\lfloor n/2 \rfloor} + N \log N)$ time this proves Theorem \ref{thm:epsBH-alg}.

\section{Center Points}

In Euclidean space a \emph{center point} $p$ of a set $X \subset \b{R}^d$ of size $N$ has the property that any halfspace that contains $p$ also contains at least $N/(d+1)$ points from $X$.  Center points always exist~\cite{Rad47} and there exists several algorithms for computing them exactly~\cite{JM94} and approximately~\cite{CEMST96,MS09}.  

We cannot directly replicate the notion of center points in $\PD n$ with horoballs.    Instead we replace it with a slightly weaker notion, which is equivalent in Euclidean space.  A \emph{horo-center point} $p$ of a set $X \subset \PD n$ (or $\b{R}^d$) of size $N$ has the property that any horoball that contains more than $N d/(d+1)$ points must contain $p$,
where we define $d = n(n+1)/2$ so that $\PD n$ is a $d$-dimensional manifold.

\paragraph{Construction for no center point in $\PD n$.}

Analogous to Euclidean center points, a center point $p$ of $X \subset \PD n$ of $N$ points has the property that any horoball that contains $p$ must also contain at least $N/(d+1)$ points from $X$.  

\begin{theorem}
For a set $X \subset \PD n$ there may be no center point.  
\end{theorem}
\begin{proof}  
Consider a set of distinct points $X \in \PD n$ such that all points $X$ lie on a single geodesic $\alpha$ between $x_1$ and $x_N$ where $x_1, x_N \in X$. Furthermore, let the points lie in a hyperbolic submanifold of $\PD{n}$.
Now, for any point $p$ not on $\alpha$ there is a horoball that contains $p$ but contains none of $X$.  So if there is a center point, it must lie on $\alpha$.  However, also for any point $p \in \alpha$ there is a horoball that intersects $\alpha$ at only $p$, since the cross-section of any horoball in the hyperbolic submanifold will be strictly convex (it can be represented as a hyperball in the Poincar\'e model, and geodesics are circular arcs, so there is a horoball tangent to the geodesic at one point).  Thus for any possible center point $p$ there is a horoball that contains at most $1$ point of $X$.  Hence, there can be no center point.  
\end{proof}

\paragraph{Horo-center points in $\PD n$.}

The ``simple'' proof of the existence of center points~\cite{Mat02} uses Helly's theorem to show that a horo-center point always exist, and then in Euclidean space a halfspace separation theorem can be used to show that a horo-center point is also a center point.  We replicate the first part in $\PD n$, but cannot replicate the second part because horoballs do not have the proper separation properties when not defined in $\b{R}^d$.  

\begin{theorem}
Any set $X \subset \PD n$ has a horo-center point.
\end{theorem}
\begin{proof}
We use the following Helly Theorem on Cartan-Hadamard manifolds (which include $\PD n$) of dimension $d$~\cite{LTZ05}.  For a family $\Eu{F}$ of closed convex sets, if any set of $d+1$ sets from $F$ contain a common point, then the intersection of all sets in $\Eu{F}$ contain a common point.  

In $\PD n$ we consider the family $\Eu{F}$ of closed convex sets defined as follows.  A set $F \in \Eu{F}$ is defined by a horofunction $b_c$ and a subset $X' \subset X$ of size greater than $N d/(d+1)$ such that $X'$ is the intersection of $X$ and a horoball $B_r(b_c)$.  Then $F = \bh{X'}$ is the ball hull of $X'$, so $F \subset B_r(b_c)$ and $F$ is compact.  

We can argue that any set of $d+1$ sets from $\Eu{F}$ must intersect.  We can count the number of points not in any $d+1$ sets as 
$$
S < \sum_{i=1}^{d+1} (N - N d/(d+1)) = \sum_{i=1}^{d+1} (N(1/(d+1)) = N.
$$
So there must be at least one point in $X$ that is in all of the $d+1$ sets.  Then by the Helly-type theorem there exists a point $p$ such that $p \in F$ for any $F \in \Eu{F}$.  

We can now show that this point $p$ must be a horo-center point.  
Any horoball that contains more than $Nd /(d+1)$ points from $X$ contains an element of $\Eu{F}$, thus it must also contain $p$.  
\end{proof}

\subsection{Algorithms for Horo-Center Points}
We provide justification for why it appears difficult to describe an exact algorithm for constructing horocenter points in $\PD n$ and then provide an algorithm for an approximate horocenter point in $\PD n$.  

Before we begin we need a useful definition of a family of problems.  An \emph{LP-type Problem}~\cite{SW92} takes as input a set of constraints $H$ and a function $\omega : 2^H \to \b{R}$ that we seek to minimize, and it has the following two properties.  
\textsc{Monotonicity:} For any $F \subseteq G \subseteq H$, $\omega(F) \leq \omega(G)$.
\textsc{Locality:} For any $F \subseteq G \subseteq H$ with $\omega(F) = \omega(G)$ and an $h \in H$ such that $\omega(G \cup h) > \omega(G)$ implies that $\omega(F \cup h) > \omega(F)$.
A \emph{basis} for an LP-type problem is a subset $B \subset H$ such that $\omega(B') < \omega(B)$ for all proper subsets $B'$ of $B$.  And we say that $B$ is a basis for a subset $G \subseteq H$ if $\omega(B) = \omega(G)$ and $B$ is a basis.  
The cardinality of the largest basis is the \emph{combinatorial dimension} of the LP-type problem.  
LP-type problems with constant combinatorial dimensions can be solved in time linear in the number of constraints~\cite{CM96}.  

\begin{lemma}
A set $H$ of horoballs in $\PD n$, and a function $\omega(G) = \min_{p \in \bigcap_{H}} \det(p)$ is an LP-type problem with constant combinatorial dimension.  
\end{lemma}

\begin{proof}
Monotonicity holds since in adding more horoballs to the a set $F \subset H$ to get a set $G \subset H$ (i.e. so $F \subset G$) we have $\bigcap_{G} \subseteq \bigcap_{F}$.  

To show locality, we consider subsets $F \subseteq G \subseteq H$ such that $\omega(F) = \omega(G)$.  Let $P$ be the set of points $\{p \in \PD n \mid \omega(p) = \min_{q \in \bigcap_{G}} \omega(q)\}$.  Adding a constraint (a horoball) $h$ to $G$ only causes $\omega(G \cup h) > \omega(G)$ if $P \cap h = \emptyset$ and thus $P \cap (\bigcap_{G \cup h}) = \emptyset$.  Since $\bigcap_G \subset \bigcap_F$, then also $P \cap (\bigcap_{F \cup h}) = \emptyset$ and $\omega(F \cup h) > \omega(F)$.  

We now show that our problem has combinatorial dimension $d = n(n+1)/2$.  
$\PD n$ is a $d$-dimensional manifold.  Each constraint (a horosphere) is a $(d-1)$-dimensional sub-manifold of $\PD n$.  Thus let $p^* = \arg \min_{p \in \bigcap_F} \omega(p)$.  If a constraint $h$ lies in the the basis $B \subset F$, then $p^*$ must lie on the corresponding horosphere.  Hence, this reduces the problem by $1$ dimension.  And each subsequent horosphere $h'$ we add to the basis, must also include $p^*$, so it must intersect $h$ (and all other horospheres in the basis) transversally, reducing the dimension by $1$.  (If $h,h' \in B$ do not intersect transversally, then we can remove either one from the basis without changing $p$.)  This process can only add $d$ horospheres to $B$ because $\PD n$ is $d$-dimensional, thus the maximum basis size is $d$.  
\end{proof}

This lemma suggests the following algorithm for constructing a horo-center point.  
Consider all subsets $X' \subset X$ of $N d/(d+1)$ points, find the minimal horoball(s) which contain $X'$.  Each of these horoballs can then be seen as a constraint for the LP-type problem.  Then we solve the LP-type problem, returning a horo-center point.  Unfortunately, there is no finite bound on the number of horoballs defined by a subset $X'$.  Theorem \ref{thm:infinite-HB} indicates that it could be infinite.  Thus there are an infinite number of constraints that may need to be considered.  

In order to approximate the horocenter point, we use a similar approach as we did to approximate the ball hull.  We discretize the set of directions, and create a finite family of constraints for each direction.  Then we can solve the associated LP-type problem to find a horo-center point.  

More formally, we place a grid $G_\eps$ on $\SO n$ so that for any $Q' \in \SO n$ there is another $Q \in G_\eps$ such that the angle between $Q'$ and $Q$ is at most $\eps/({n \choose 2} 2 \sqrt{2} \sinh(d_X/\sqrt{2}))$.  For each $c$ corresponding to $Q \in G_\eps$, we consider $\pi_F(X)$, the projection of $X$ onto the $(d-1)$-flat $F$ corresponding to $Q$.  Within $F$, we can consider all subsets $X' \subset X$of $Nd/(d+1)$ points and find the hyperplanes defining the convex hull of $\pi_F(X')$.  This finite set of hyperplanes corresponds to a finite set of horoballs which serve as constraints for the LP-type problem in $\PD n$.

We say a point $\hat p$ is an \emph{$\eps$-approximate horo-center point} if there is a horo-center point $p$ such that for any horofunction $b_c$ we have $|b_c(p) - b_c(\hat p)| \leq \eps$.  

\begin{lemma}
A point $\hat p$ that satisfies all of the constraints defined by $X$ and $G_\eps$ is an $\eps$-approximate horo-center point of $X$.
\end{lemma}
\begin{proof}
Let $C(X) \subset \PD n$ be the set of horo-center points.
Assume that $\hat p \notin C(X)$, otherwise let $p = \hat p$ and we are done.  

We show the existence of a specific nearby horo-center point $p$ with the following property.  Let $c$ be the geodesic connecting $\hat p$ and $p$, and let $p = \max_{p' \in C(X)} b_c(p')$.  
For any $q \in \bdry{C}(X)$ let $\alpha$ be the geodesic connecting $q$ and $\hat{p}$.  If $q = \arg \max_{p' \in C(X)} b_\alpha(p')$, we are done, if not, let $q_\alpha \in \bdry{C(X)}$ such that $q_\alpha = \arg \max_{p' \in C(X)} b_\alpha(p')$.  Then the geodesic ray on $C(X)$ that goes from $q$ to $q_\alpha$ describes a flow on $C(X)$.  We can see that the fixed point of this flow is $p$, because as we move to $\bar q$ in the direction of this flow, $b_\alpha$ gets closer to is maximum, and the geodesic on $\PD n$ from $\hat p$ to $\bar q$ is closer to direction defining the horofunction that $\bar q$ maximizes.  

Now, we can show that $b_c(\hat p) - b_c(p) = \delta \leq \eps$.  This follows by Theorem \ref{lem:Lipschitz-gen} since there must be another direction $Q \in G_\eps$ where the corresponding flat contains a geodesic $c'$ such that there are more than $N d/(d+1)$ points $x \in X$ such that $b_{c'}(x) < b_{c'}(\hat p)$ and $|b_c(x) - b_{c'}(x)| \leq \eps$.  Thus there are more than $N d/(d+1)$ points $x \in X$ such that $b_c(x) -\eps < b_c(\hat p)$, and there must be exactly $\lfloor N d/(d+1) + 1\rfloor$ points $x \in X$ such that $b_c(x) < b_c(p)$.  Since $\hat p$ and $p$ lie on the geodesic $c$, we have $\delta \leq \eps$.  

We can now use Lemma \ref{lem:nabla} to bound the difference in values for any horofunction.  $||\nabla b_c||$ is constant for any $b_c$, hence for any other horofunction $b_{c'}$ we have $|b_{c'}(\hat p) - b_{c'}(p)| \leq |b_c(\hat p) - b_c(p)| \leq \eps$ (where they are only equal when $c$ and $c'$ asymptote at opposite points).  Since $p$ is a horo-center point, this concludes the proof.  
\end{proof}

When constructing the set of constraints in each flat $F$ corresponding to a direction in $G_\eps$ we do not need to explicitly consider all ${N \choose Nd/(d+1)}$ subsets of $X$.  Each constraint only depends on $n$ points in $X$, so we can instead consider ${N \choose n} = O(N^{n})$ subsets of $X$ of size $n$, and then check if either of the halfspaces it defines in $F$ contain at least $Nd/(d+1)$ points from $X$ in $O(N)$ time.  Only these constraints need to be considered in the definition of $\hat p$.

\begin{theorem}
Given a set $X \subset \PD n$ of size $N$, (for $n$ constant) we can construct an $\eps$-approximate horo-center point in time $O((\sinh(d_X)/\eps)^{n-1} N^{n+1})$ time.  
\label{thm:apx-hc-point}
\end{theorem}
\begin{proof}
As per the above construction, for each $Q \in G_\eps$ we only need to consider $O(N^{n})$ potential constraints, and each takes $O(N)$ time to evaluate.  
By Lemma \ref{lem:Lipschitz} $G_\eps$ is of size $O((\sinh(d_X)/\eps)^{n-1})$ so the total number of constraints we need to consider in the LP-type problem is $O((\sinh(d_X)/\eps)^{n-1} N^{n})$.  The total runtime is thus dominated by constructing the constraints and takes $O((\sinh(d_X)/\eps)^{n-1} N^{n+1})$ time.   
\end{proof}

\newpage
\bibliography{refs,mrefs-short}
\bibliographystyle{acm}

\appendix

\end{document}